\documentclass[11pt]{article}

\usepackage{amsfonts,amsmath,amssymb,amsthm}
\usepackage{graphicx} 
\usepackage{float}
\usepackage{caption}
\usepackage{subcaption}
\usepackage{graphicx,comment}
\usepackage[export]{adjustbox}
\usepackage{fullpage}
\usepackage[margin=0.5in]{geometry}
\usepackage{xcolor}
\usepackage{hyperref}
\usepackage[capitalize]{cleveref}
\usepackage{authblk}

\newtheorem{thm}{Theorem}
\newtheorem{defn}{Definition}

\newtheorem{rem}{Remark}
\newtheorem{lemma}{Lemma}
\newtheorem{corollary}{Corollary}

\newenvironment{myproof}{
  \par\medskip\noindent
  \textit{Proof}.
}{
\newline
\rightline{$\qedsymbol$}
}

\def\eps{\varepsilon}

\def\Pr{\mathbb{P}}

\def\AM{\operatorname{AM}}
\def\AR{\operatorname{AR}}

\def\SI{Supplementary Information}

\def\mor{\mathcal{M}}
\def\sen{\mathcal{R}}

\def\eps{\varepsilon}

\def\fp{\rho_r^{\mor}} %
\def\fpn{\rho_{r=1}^{\mor}}

\def\fps{\rho_r^{\sen}} %
\def\fpns{\rho_{r=1}^{\sen}}

\def\Sq{\operatorname{Sq}}

\def\maintitle{Replacers and their evolutionary stability in the Moran process on graphs}

\usepackage{lineno}

\author[1]{Michal Pecho}
\author[2]{Jakub Svoboda}
\author[3]{Lenka Kopfová}
\author[1]{Josef Tkadlec}
\author[3]{Krishnendu Chatterjee}

\affil[1]{Computer Science Institute, Charles University, Prague, Czech Republic}
\affil[2]{Dartmouth College, Hanover NH, USA}
\affil[3]{Institute of Science and Technology, Austria}

\begin{document}

\title{\maintitle}

\date{}
\maketitle

\begin{abstract}
Evolutionary dynamics in finite structured populations are commonly modeled by the Moran Birth-death process.
In each discrete step, a single random individual is selected for reproduction, and then its offspring replaces a random neighbor.
A key quantity is the fixation probability of a single invader attempting to take over a population of residents.
A recent work introduced a new neighborhood-aware phenotype called a replacer.
A replacer never wastes their reproductive turn by always replacing an individual of the other type (if available).
In this work, we study the evolutionary stability of resident replacers who are invaded by mutant replacers. 
We find that residents are strongly protected against such invasions, and we quantify the strength of this effect by showing three types of results.
First, we show that on well-mixed populations of size $N$, the invader fixation probability is exponentially small in $N$, even when the invader has a fixed relative reproductive rate $r>1$, and the same holds for all high-degree graphs.
Second, we study bounded-degree graphs. We prove that on cycles, the fixation probability of an advantageous invader decreases only by a constant factor.
However, we also present graphs with maximum degree 4, where the invader fixation probability is exponentially small in $N$ whenever $r\le1.9$. %
Thus, high degrees are sufficient for evolutionary stability, whereas with low degrees the evolutionary stability depends on specific features of the underlying spatial structure.
Third, we prove general bounds for arbitrary graphs.
Namely, we show that for any graph $G$ the invader fixation probability drops below the natural baseline given by the standard Moran process with oblivious individuals on a well-mixed population, both for $r\approx 1$ and for $r\ge 2$.
Together, our results establish that the evolutionary dynamics of replacers is better characterized by the phrase ``survival of the first'' rather than the classic ``survival of the fittest''.
\end{abstract}

\section*{Introduction}

Evolution in finite populations is shaped by the interaction of mutation and selection.
Mutation generates new types (variants), while natural selection prunes the existing types by favoring individuals with higher fitness.
Evolution is further shaped by the spatial structure of the population and by the inherent randomness~\cite{ewens2004mathematical,durrett1994importance}.
A common stochastic description of the evolutionary dynamics on a spatial structure is the Moran Birth-death process~\cite{moran1958random,lieberman2005evolutionary,broom2022game}.
There are $N$ individuals, spread over $N$ nodes of a given graph (network) $G_N$, one individual per node.
The nodes of the graph represent possible sites, and the edges (links) represent where the offspring can migrate.
An unstructured, well-mixed population is represented by the complete graph $K_N$ with all edges present, but graphs can represent any spatial structures, including island models or regular lattices of any dimension.
In each step of the Moran process, first an individual is selected with probability proportional to its reproductive rate, and then its offspring replaces another individual chosen uniformly at random from among its neighbors in the graph, see~\cref{fig:model}a.

\begin{figure}[!hbt]
    \centering
    \includegraphics[width=0.9\linewidth]{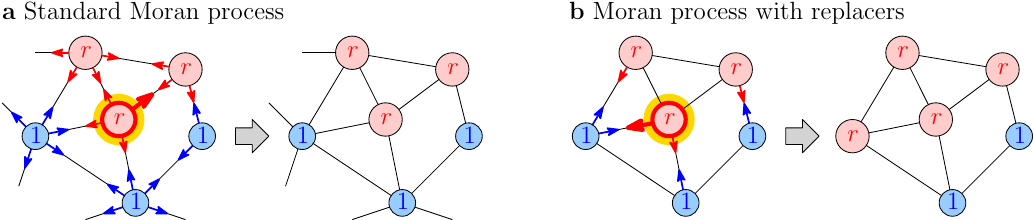}
        
        \caption{
        \textbf{Moran process with oblivious individuals and with replacers.} The population consists of invaders (red, reproductive rate $r$) and residents (blue, reproductive rate $1$) spread over nodes of a graph (network). In each step of the Moran Birth-death process, one individual produces an offspring that replaces a random neighbor.
        \textbf{a,}~In the standard Moran process, individuals are oblivious, so the offspring of a reproducing individual (highlighted) could replace any neighbor (arrows), possibly wasting its reproductive turn (thick red arrow).
        \textbf{b,} In the Moran process with replacers, the offspring replaces a random neighbor only from among those that have a different type, if available.
        }\label{fig:model}
\end{figure}

When mutations are rare, a central quantity is the probability that a single invader with relative reproductive rate $r$ produces a lineage that eventually takes over the entire population of $N-1$ residents, each with reproductive rate $1$~\cite{kimura1968evolutionary,kimura1962probability,desai2007speed}.
We denote this fixation probability by $\fp(G_N)$, where $\mor$ stands for Moran.
When the mutation is neutral ($r=1$), the fixation probability is equal to $1/N$ on any spatial structure $G_N$. %
For non-neutral mutants ($r\ne1$), the situation is more complicated.
On a well-mixed population $K_N$, the fixation probability is given by a formula
\[
    \fp(K_N)=
    \frac{1-\frac1r}{1-\frac1{r^N}}.
\]
That is, for fixed $r>1$ and large $N$ the fixation probability approaches a positive constant $1-1/r$, whereas for fixed $r<1$ and large $N$ it is exponentially small.
Thus, in large populations an advantageous mutant ($r>1$) is somewhat likely to fixate, whereas disadvantageous mutants ($r<1$) are extremely unlikely to fixate.
In fact, the same formula holds for a class of so-called isothermal structures, represented by regular graphs (all nodes have the same number of neighbors).
Thus, the formula serves as a baseline for evaluation of the effects of a spatial structure.

Relative to this baseline, some population structures, known as amplifiers, increase the fixation probability of a mutant in the original process~\cite{adlam2015amplifiers,galanis2017amplifiers,tkadlec2021fast,svoboda2024amplifiers}.
The strongest known amplifiers can make the fixation probability converge to $1$ as the population size grows, for every fixed mutant fitness $r>1$. 
On the other hand, suppressors reduce the fixation probability of an advantageous mutant~\cite{giakkoupis2016amplifiers}, and thus they help maintain evolutionary stability.
By monotonicity~\cite{diaz2016absorption}, the fixation probability of any advantageous mutant ($r>1$) on any spatial structure is at least~$\frac{1}{N}$. See~\cite{diaz2021survey} for a survey.

The individuals in the Moran process are oblivious in the sense that occasionally they waste their reproductive turn by replacing a neighbor of the same type as the parent.
A recent work introduced the notion of a replacer phenotype~\cite{pecho2025selective}.
A replacer is an individual who never wastes their reproductive turn, by replacing a random neighboring individual from among only those that are of a different type than the parent (if such neighbors are available), see~\cref{fig:model}b.
It is speculated that the replacer phenotype occurs in settings as diverse as cancer evolution in Drosophila~\cite{moreno2004dmyc}, or hyphal interference in funghi~\cite{boddy2016fungal}.
The authors of~\cite{pecho2025selective} considered the evolutionary dynamics between replacers and the standard, oblivious individuals.
For example, they found that when a neutral replacer ($r=1$) invades a population of oblivious residents, the fixation probability substantially increases to roughly $1/\sqrt N$, from the baseline value of $1/N$.
However, for advantageous replacers on the well-mixed population, the fixation probability remains unchanged at roughly $1-1/r$.

In this work, we study the evolutionary dynamics between replacers and replacers.
Thus, we consider a single replacer invader, as it attempts to take over a homogeneous background population of replacer residents.
This situation is likely to occur once replacers become established in the population for the first time.
Our main question is to analyze how robust are the established replacers with respect to subsequent replacer invasions.
To that end, we denote by $\fps(G_N)$ the fixation probability of a single replacer invader with relative reproductive rate $r$ (here $\sen$ stands for replacers). %
We present three types of results. Together, they establish that replacers are generally well protected against further replacer invasions.

First, we consider the well-mixed population, represented by the complete graph $K_N$.
We derive an explicit formula for the fixation probability of an invading replacer with any reproductive rate $r$, see~\cref{thm:kn}.
It turns out that the fixation probability is exponentially small in $N$, even when the invader has a fixed advantage $r>1$.
This is in stark contrast with the baseline setting of oblivious individuals, where an invader with an advantage $r>1$ has a constant chance of taking over.
Thus, replacers in large well-mixed populations are extremely robust with respect to subsequent replacer invasions of any strength.
Moreover we show that this effect is present in all graphs where the minimum degree is large.

Second, we study regular bounded-degree graphs such as lattices in 1 or 2 dimensions.
For the 1-dimensional lattice, represented by the cycle graph $C_N$, we again derive an explicit formula, see~\cref{thm:cn_main}.
In particular, the fixation probability of a neutral replacer ($r=1$) becomes $1/(2N-2)$, so roughly a half of the value $1/N$ given by the baseline setting of oblivious individuals.
For advantageous mutants ($r>1$), the drop is by a factor of $1+1/r$.
For 2-dimensional lattice, the drop in the fixation probability is more pronounced.
Moreover, there exist graphs with maximum degree 4, where the fixation probability is exponentially small in $N$ whenever $r\le 1.9$. Thus, even on regular bounded-degree graphs the increase in the stability of established replacers could be extreme.

Third, we prove general bounds for arbitrary graphs.
Namely, we show that for any graph $G$ on $N\ge 3$ nodes, the invader fixation probability drops strictly below the natural baseline value of the setting with all individuals oblivious, both for $r\approx 1$ and for $r\ge2$.
Together, our results establish that the evolutionary dynamics of replacers is better characterized by a phrase ``survival of the first'' rather than the classic ``survival of the fittest''.

\section{Model}

The population size is denoted by $N$.
The spatial structure is represented by a graph $G = (V,E)$, where the set $V$ of $N$ nodes represents the possible sites.
Each site is occupied by a single individual.
The set $E$ of edges represents there the offspring can migrate.
We denote the neighborhood of a node $v$ by $N(v) = \{u\mid \{u,v\} \in E\}$.
The number of neighbors of a node $v$ is called its degree, and is denoted by $d_v$.
The minimal degree of a graph is denoted $d_{\min}$ and the maximal degree is denoted $d_{\max}$.
Graphs with high $d_{\min}$ are called \emph{high-degree} graphs and graphs with low $d_{\max}$ are called \emph{low-degree graphs}.
We focus on regular graphs, these are graphs for which $d_{\min} = d_{\max}$.

The evolutionary process plays out on graph $G=(V,E)$.
Each node is occupied by one individual, either a resident with reproductive rate $1$ or a mutant with reproductive rate $r>0$.
We consider discrete-time birth--death updating.
In one step, one individual is selected randomly, proportionally to its reproductive rate to reproduce ($r/F$ for mutant and $1/F$ for resident, where $F$ is the total reproductive rate in the graph).
The offspring inherits the type of its parent and replaces an individual at some neighboring node.
In the \emph{standard Moran process}, denoted $\mor$, the offspring replaces a random neighbor of the parent.
In the \emph{Moran process with replacers}, denoted $\sen$, the offspring replaces a random neighbor of an opposite type of the parent (if at least one such neighbor is available).
If all neighbors have the same type as the parent, then the reproduction is wasted and nothing changes.

Our key quantity is the fixation probability of a single replacer, placed at a random node of the graph $G_N$, who invades a background population of replacer residents.
We denote this fixation probability by $\fps(G_N)$.
For comparison, we denote by $\fp(G_N)$ the analogous fixation probability in the standard Moran process, that is, when both the invader and the replacers are oblivious.

\section{Results}

We present three types of results.
First, we study the well-mixed populations and high-degree graphs where each node has many neighbors.
Second, we cover the opposite end by studying low-degree graphs such as lattices.
Third, we present general bounds for arbitrary graphs.

\subsection{High-degree graphs}
First, we consider the complete graphs $K_N$ which model well-mixed populations.
The state of the evolutionary dynamics on $K_N$ can be represented by the number $k\in\{0,1,2,\dots,N\}$ of invaders, where $k=N$ corresponds to invaders successfully taking over the population and $k=0$ corresponds to invaders becoming extinct. Initially, we have $k=1$.
For each $k\in\{1,2,\dots,N-1\}$ we denote by $p_k$ the probability that in the next step of the evolutionary process the number of invaders increases to $k+1$.
Likewise, we denote by $q_k$ the probability that it decreases to $k-1$.
A direct calculation yields
\[
p_k = \frac{k\cdot r}{k\cdot r + (N-k)\cdot 1}\cdot 1,
\]
since $\frac{k\cdot r}{k\cdot r + (N-k)}$ is the probability that one of the $k$ invaders is selected for reproduction and once it is selected, it always replaces a resident. Similarly, we obtain
\[q_k = \frac{(N-k)\cdot 1}{k\cdot r + (N-k)\cdot 1}\cdot 1.
\]
Thus, the backward bias $\gamma_k$ defined as $\gamma_k=\frac{q_k}{p_k}$ satisfies $\gamma_k=\frac{N-k}{k\cdot r}$. 

Note that for small $k$ (namely $k < N/(1+r)$) the backward bias is greater than 1, thus a successful invasion is unlikely.
Applying known results on absorption probabilities of 1-dimensional Markov chains~\cite{nowak2007evolutionary} we are then able to derive a formula for the invader fixation probability $\fps(K_N)$ on the complete graph $K_N$.
Moreover, we extend this argument to obtain a bound on the invader fixation probability $\fps(K_N)$ for any graph $G_N$ where each node has degree at least $D$.

\begin{thm}\label{thm:kn}
\leavevmode
\begin{enumerate}
    \item (Complete graph) Let $K_N$ be a complete graph and $r>0$. Then
\[
\fps(K_N) = \frac{1}{\left(1+\frac1r\right)^{N-1}}.
\]
\item (High-degree graphs) Let $G_N$ be any graph with minimum degree $D$. Then
\[
\fps(G_N) \leq \frac{1}{(1+\frac1r)^D}.
\]
\end{enumerate}
\end{thm}

Note that the invader fixation probability is exponentially small in $N$, for every constant $r$.
Thus, even if the invading replacer has a constant advantage ($r>1$), it is exponentially unlikely to successfully invade an established population of replacer residents. 
Recall that in the standard Moran process with oblivious invaders and oblivious residents we have $\fp(K_N) = \frac{1-1/r}{1-1/r^{N}}$, thus advantaged invaders have a constant chance of roughly $1-\frac1r$ of taking over.
Therefore, established replacers are much better protected against replacer invasions than oblivious individuals against oblivious invasions.
This remains true even when there are initially $i>1$ invaders, as long as $i$ is sufficiently small, see~\cref{fig:kn} for an illustration and \SI{} for details.

\begin{figure}[!hbt]
    \centering
    \includegraphics[width=0.14\linewidth,page=2,valign=c]{figures/fig_graphs.pdf}
    \includegraphics[width=0.42\linewidth,valign=c]{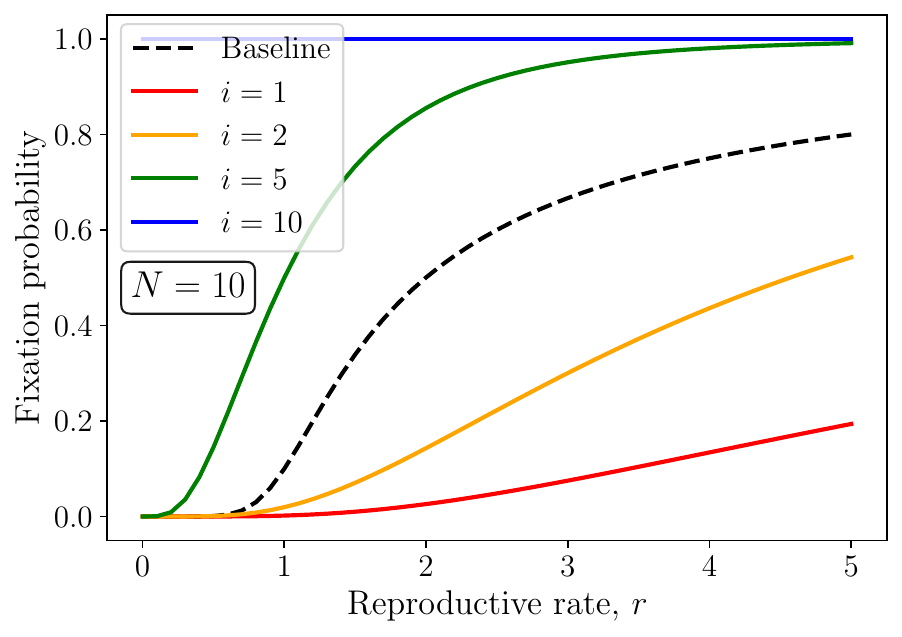}
    \includegraphics[width=0.42\linewidth,valign=c]{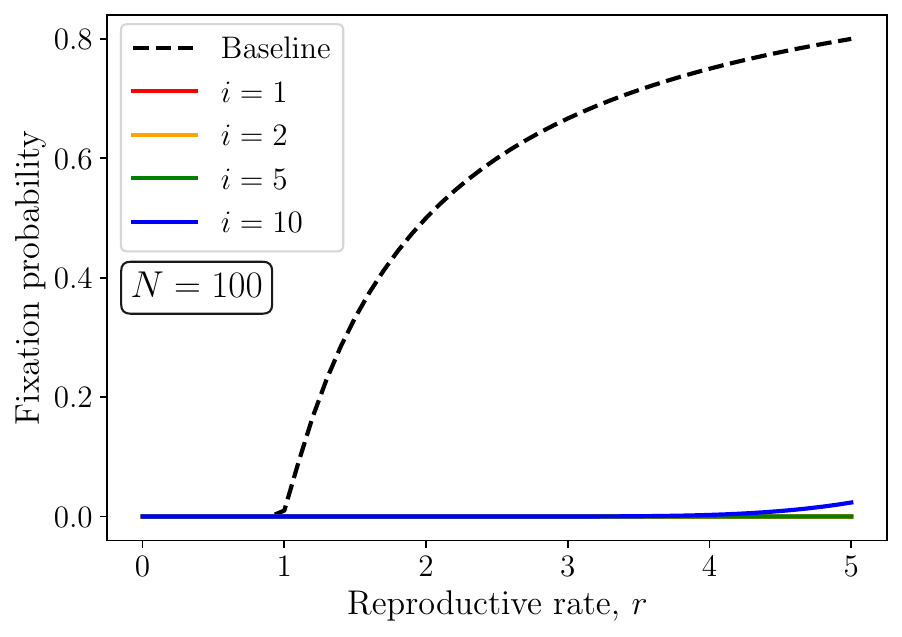}
        \caption{
        \textbf{Evolutionary stability of well-mixed populations.}
        \textbf{a,} A well-mixed population of size $N$ is represented by a complete graph $K_N$.
        \textbf{b,}
        The fixation probability $\fps(K_{10})$ of a single replacer with relative reproductive rate $r$ %
        who is invading a well-mixed population $K_{10}$ of replacers (red line) is substantially smaller than the corresponding fixation probability for an oblivious mutant invading an oblivious population (dashed line). For comparison, we also plot the fixation probabilities when there are initially $i\in\{1,2,5,10\}$ replacer invaders (colors). Here $N=10$.
        \textbf{c,} With population size $N=100$, the difference is even more pronounced -- large well-mixed populations $K_N$ are extremely robust with respect to replacer invasions with any fixed relative reproductive rate $r$, see~\cref{thm:kn}.
        The plotted values are obtained from a formula given by Lemma~2 %
        in \SI{}.
        }\label{fig:kn}
\end{figure}

The intuition behind the first part of~\cref{thm:kn} is that due to the high connectivity of the complete graph, the initial invader is likely to immediately disappear.
When $N$ is large, the individual selected for the first reproductive step will typically be one of the established residents. Since each of those individuals has the invader as a neighbor, the neighbor gets immediately eliminated and the invasion attempt fails.
More careful accounting yields that the invader fixation probability is even exponentially small in $N$.

In fact, as stated in the second part of~\cref{thm:kn}, an analogous statement holds for any high-degree graph.
Consider a fixed $r>1$ and any graph $G_N$, where each node has degree at least $D\ge \frac1{100}N$, that is, roughly speaking each node is connected to at least 1 \% of all the other nodes.
Then the replacer fixation probability satisfies $\fps(G_N)\le 1/c^N$, where $c=(1+\frac1r)^{1/100}>1$ is a fixed constant. In particular, the fixation probability is exponentially small in the population size $N$, see \SI{} for details.
Given this intuition, it is natural to investigate graphs where nodes have few neighbors.

\subsection{Low-degree graphs}
A complete graph is a regular graph where each node has the largest possible degree.
At the opposite extreme, there exists a unique connected 2-regular graph, namely the cycle graph $C_N$, that models a 1-dimensional lattice with a periodic boundary condition.
Similarly, a 2-dimensional lattice is represented by a square grid graph $\Sq_N$ (with periodic boundary condition) where each node has 4 neighbors. Other 4-regular graphs exist, for example the Bead graph $B_N$, see~\cref{fig:fp_vs_r}.

\begin{figure}[!hbt]
    \centering
\includegraphics[width=0.4\linewidth]{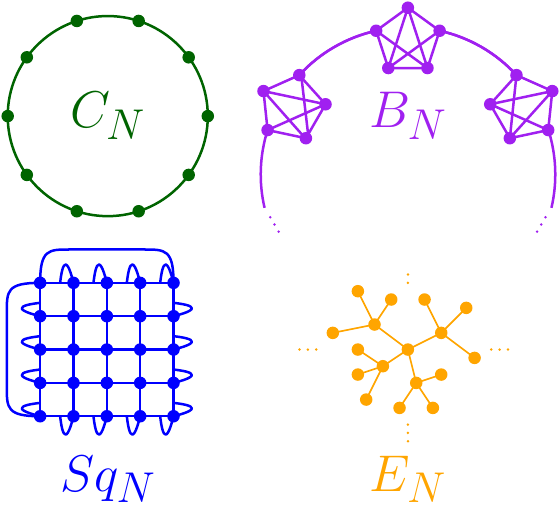}\qquad 
\includegraphics[width=0.45\linewidth]{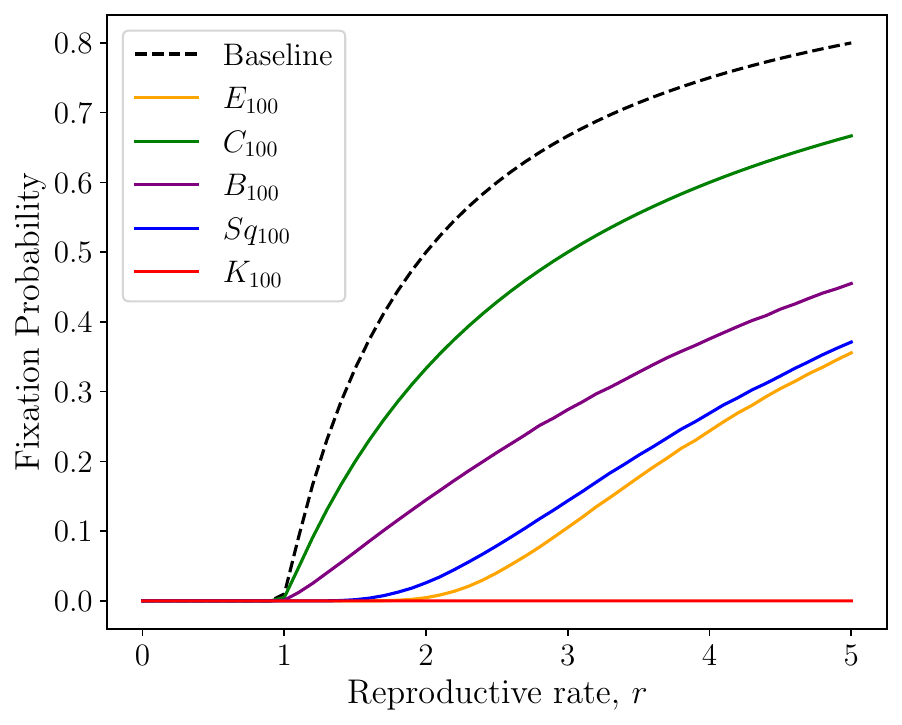}
        \caption{
        \textbf{Evolutionary stability on low-degree graphs.}
        \textbf{a,} In a cycle graph $C_N$ (green), each node has two neighbors. In a grid graph $\Sq_N$ with periodic boundary condition (blue), in a bead graph $B_{N}$ (purple), and in a 4-regular expander graph $E_N$ (orange), each node has 4 neighbors.
        \textbf{b,}
        Fixation probability $\fps(G_{N})$ of a single replacer invader ($y$-axis) crucially depends both on the invader reproductive rate $r$ ($x$-axis) and on the spatial structure $G_N$ (colors).
        Generally, sparser graphs such as the cycle $C_N$ lead to higher fixation probabilities, but even on equally dense graphs the fixation probability is sensitive to features of the graph.
        In particular, for different 4-regular graphs on $N=100$ nodes with $r=1.5$ we have $\fps(\Sq_{100})<0.2\%$, whereas $\fps(B_{100})>7\%$.
        Complete graph $K_N$ and the natural baseline corresponding to the oblivious individuals shown for comparison.
        Values for $\Sq_N$, $B_N$, and $E_N$ computed by $10^6$ simulations for $r\in\{0,0.1,0.2,\dots,5\}$. For $E_N$, we consider 5 random 4-regular graphs and plot the values for the one graph that minimizes $\rho^{\sen}_{r=2}(G_{100})$. 
        Values for $C_N$ and $K_N$ given by the formulas in Theorems~1 and~2. %
        }\label{fig:fp_vs_r}
\end{figure}

The evolutionary dynamics of replacers on such low-degree graphs is very different from the dynamics on the complete graph, see~\cref{fig:fp_vs_r}.
While replacer fixation probability on large complete graphs remains exponentially small even when the invader is advantageous ($r>1$), on cycle graphs $C_N$ the fixation probability is substantial already for $r$ marginally larger than $1$.
For 4-regular graphs, the replacer fixation probability depends on specific features of the graph.
On some graphs such as the bead graph $B_N$, it becomes substantial already for $r$ marginally larger than 1, whereas on other graphs it remains negligible for reproductive rates as high as roughly 2.
For even larger population sizes, we can prove the following results about the cycle graphs and about different 4-regular graphs.

\begin{thm}\label{thm:cn_main}
\leavevmode
\begin{enumerate}
    \item (Cycle graph) Let $C_N$ be a cycle graph. Then    %
    \begin{enumerate}
        \item %
         If $r=1$ then $\fps(C_N)=\frac1{2N-2}$.
        \item %
         If $r>1$ and $N$ is large then
        $\fps(C_N) \approx\frac{1-\frac1r}{1+\frac1r}$.
        \item %
         If $r<1$ and $N$ is large then $\fps(C_N)\sim r^N$.
    \end{enumerate}
    
    \item (4-regular graphs)
    \begin{enumerate}
        \item For every $r\ge 1.1$ there exists a positive constant $c$ and large 4-regular graphs $B_N$ such that $\fps(B_N)>c$.
        \item For every $r\le 1.9$ there exists a positive constant $c>1$ and large 4-regular graphs $E_N$ such that $\fps(E_N)<1/c^N$.
    \end{enumerate}
    
    \end{enumerate}
\end{thm}

The first part of~\cref{thm:cn_main} implies that, compared to the oblivious baseline $\fp(C_N)=1/N$, the fixation probability of neutral replacer invaders ($r=1$) is roughly halved.
For advantageous invaders and large $N$, the oblivious baseline is given by $\fp(C_N)\approx 1-1/r$, so the fixation probability of replacer invaders drops by a factor of roughly $1+1/r$.
For disadvantageous invaders ($r<1$), both $\fp(C_N)$ and $\fps(C_N)$ are exponentially small in $N$.
Therefore, established replacers are better protected against replacer invasions than oblivious individuals against oblivious invasions, but the effect is weaker than on the complete graph.
In particular, replacers with any advantage $r>1$ are reasonably likely to successfully invade on a cycle graph, whereas they are exponentially unlikely to successfully invade on the complete graph.

The idea behind the evolutionary dynamics on cycles is that, as with the complete graphs, the current state of the population on a cycle graph can be represented by the number $k$ of invaders (who always form a contiguous block along the cycle), so we again obtain a 1-dimensional Markov chain.
A calculation similar to the case of the well-mixed population yields
the backward biases $\gamma_k= \frac{1}{r}$ for $2 \le k \le N-1$.
This allows us to derive an exact formula for the fixation probability on a cycle graph, for any $r$ and any $N$.
Analyzing the formula in three regimes $r=1$, $r>1$, $r<1$ gives the results presented in the first part of~\cref{thm:cn_main}, see \SI{} for details.

The second part of~\cref{thm:cn_main} shows that the effects observed in~\cref{fig:fp_vs_r} for $N=100$ remain in place even for larger population sizes $N$.
Namely, for some 4-regular graphs such as the bead graphs $B_N$ the replacer fixation probability becomes subst,antial already for invader advantages $r$ that are only mildly larger than $1$.
In contrast, for other 4-regular graphs, the fixation probability remains exponentially small for a range of $r$ values up to at least $r\le 1.9$.

The idea behind the analysis of the bead graphs is to show that during the evolutionary dynamics, the individual beads are typically homogeneous (either all occupied by invaders or all occupied by residents). And since the invaders have a small advantage $r>1$, they are more likely to take over the next bead along the chain, as compared to losing a bead.
The dynamics is thus not too different from the dynamics on the cycle graph.

In contrast, on other 4-regular graphs the interface between invaders and residents is more complicated.
Typically, early on in the process there are fewer invaders than residents along this interface, which results in the invader population being more likely to diminish than to increase.
This bias against invaders then makes their fixation probability exponentially small, even if their relative advantage is as high as $r=1.9$.
In fact, this situation occurs whenever the underlying graph is a relatively good expander~\cite{alon1986eigenvalues}.
Since most 4-regular graphs indeed are relatively good expanders~\cite{expanders}, this behavior is typical in the sense that it occurs for at least 50 \% of large randomly generated 4-regular graphs, see \SI{} for details.

We conclude this section with two more comments.
First, our last result generalizes to low-degree graphs with higher connectivity.
For example, in a 3-dimensional cubic lattice each node is adjacent to 6 other nodes, and we can show that for 50 \% of 6-regular graphs, the replacer fixation probability is exponentially small, as long as the invader has relative reproductive rate $r\le 3.9$.
In general, on $d$-regular graphs with $d\ge 4$ the replacers are typically extremely well protected against invasions with reproductive rates as high as $r\le d-2.1$. See \SI{} for details.

Second, we note that square grids $\Sq_N$ are actually not that good expanders, and in~\cref{fig:fp_vs_r} we indeed observe that $\fps(\Sq_{100})$ becomes substantial already for $r<1.9$. This points at the expansion property of the underlying graph as a critical features that affects the replacer dynamics.

\subsection{Arbitrary graphs}
As our third main contribution, we provide bounds on the evolutionary stability with respect to replacer invasions that apply to any spatial structure, whether regular or not.

Recall that in the standard Moran process with oblivious individuals, the spatial structures represented by regular graphs are well-behaved in the sense that the fixation  probability of a single (oblivious) invader with relative reproductive rate $r\ne 1$ is given by the formula $\frac{1-1/r}{1-1/r^N}$. 
For non-regular graphs, the fixation probability is generally different, and it can be both higher and lower than the baseline given by the regular graphs~\cite{lieberman2005evolutionary}.
For small population sizes $N$, it is possible to numerically compute features of the evolutionary dynamics~\cite{hindersin2019computation,hindersin2015most,moller2019exploring}.
Here we compute the fixation probability for a single replacer invading a population of replacers, for each of the 853 connected graphs on $N=7$ vertices, see~\cref{fig:all_vertices}.

\begin{figure}[!hbt]
    \centering
    \includegraphics[width=0.45\textwidth]{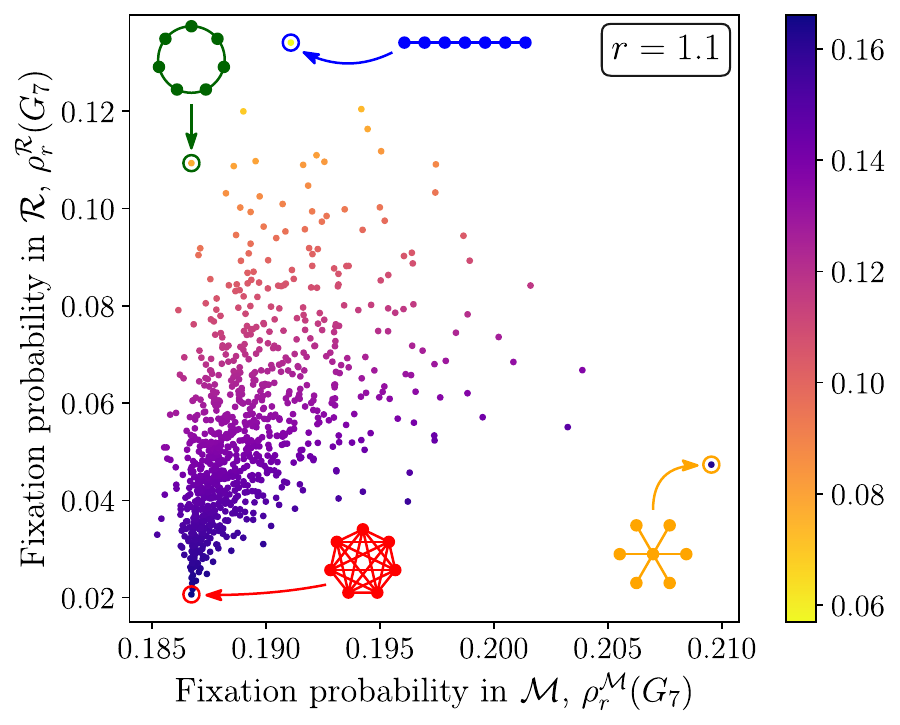}
    \includegraphics[width=0.45\linewidth]{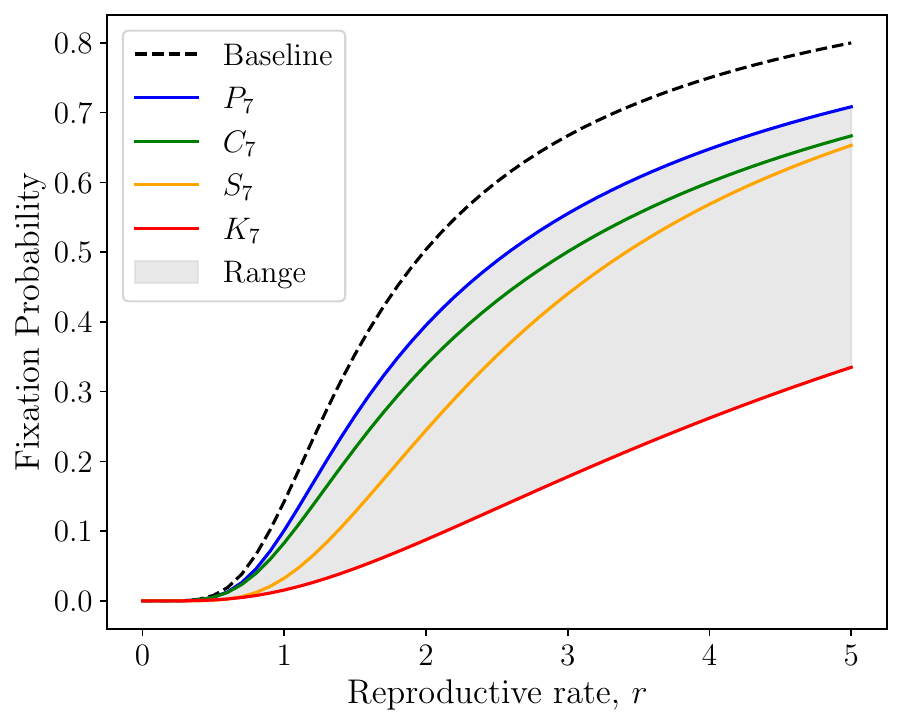}
        \caption{
        \textbf{All spatial structures of size $N=7$.}
        \textbf{a,} For each of the 853 connected spatial structures with population size $N=7$, we compute the fixation probability of one randomly placed mutant in the standard Moran process ($x$-coordinate) and in the Moran process with replacers ($y$-coordinate). The color of each dot in the scatter plot represents the difference between the two probabilities.
        The difference is always positive, indicating that replacers are more robust with respect to replacer invasions compared to oblivious individuals with respect to oblivious invasions.
        With replacer invasions, the complete graph is the most robust.
        Interesting outlier graphs are highlighted. Here $r=1.1$, %
        see~\SI{} for $r\in\{0.9,2\}$.
        \textbf{b,} For each $r\in\{0,0.1,0.2,\dots,5.0\}$, the replacer fixation probability $\fps(G_7)$ is maximized when $G_7=P_7$ is the path graph (blue), and minimized when $G_7=K_7$ is the complete graph (red).
        For reference, we also plot fixation probabilities for the cycle graph $C_7$, the star graph $S_7$, and the natural baseline for the oblivious case.
        }\label{fig:all_vertices}
\end{figure}

We find that all fixation probabilities lie below the natural baseline given by the regular graphs in the oblivious case.
Thus, established replacers are more robust with respect to replacer invasions than the natural baseline.
Moreover, among the connected graphs on $N=7$ nodes, the fixation probability is lowest for the complete graph $K_7$, and highest for the path graph $P_7$.

Repeating the calculation for larger $N$ quickly becomes infeasible, since already for %
$N=15$ there are more than 10 quintillion ($10^{19}$) 
connected graphs~\cite{oeis}.
Nevertheless, we are able to mathematically prove two results that apply to any spatial structure of any size.
They show that established replacers are more robust with respect to replacer invasions than the natural baseline, even for large population sizes $N$, provided that the invader is near-neutral or sufficiently strong.

\newpage

\begin{thm}\label{thm:arbitrary_main}
\leavevmode
\begin{enumerate}
    \item ($r= 1$)
    Let $G_N$ be any graph with $N\geq 3$ nodes. Then $\fpns(G_N)<\frac{1}{N}=\fpn(G_N)$.
    \item ($r\ge 2$) Let $G_N$ be any graph with $N\ge 2$ nodes and $v$ a node of $G_N$. Let $r \geq 2$. Then
    \[
        \fps(G_N, v) \le \frac{1-\frac{1}{r}}{1-\frac{1}{r^N}}\,.
    \]
    Moreover, the equality occurs if and only if $G_N$ is the path graph and $v$ is one of its endpoints.
    \end{enumerate}
\end{thm}

We now discuss the two claims separately.

Regarding the first claim, note that the assumption $N\ge 3$ is necessary, since for $N=2$ there is only one connected graph $G_2$ and we have $\fpns(G_2)=\fpn(G_2)=\frac12$.
Also, note that since the inequality $\fpns(G_N)<\frac{1}{N}$ is strict and the fixation probability on any graph is a continuous function of the invader reproductive rate $r$, an analogous claim can be made for near-neutral invaders. That is,
for every graph $G_N$ with $N\geq 3$ nodes there exists $\varepsilon >0$ such that for every $r \in (1-\varepsilon,1+\varepsilon)$ we have both $\fps(G_N) < \frac1N$ and $\fps(G_N) < \fp(G_N)$. 

The idea behind the proof of the first claim is to use a technique called coupling~\cite{lindvall2002lectures}. Instead of the Moran process with a single replacer invading a homogeneous population of replacers, we analyze a different process where the established population consists of $N-1$ mutually different replacers.
The proof then has two components.
First, we show that the invader fixation probability in this new process is strictly higher than in the homogeneous case.
Intuitively, this is because the established replacers now compete not only with the invader, but also with each other.
Second, using a symmetry argument we show that the invader fixation probability in the new process is equal to exactly $1/N$.
See \SI{} for details.

Next, we discuss the second claim that concerns invaders with relative reproductive rate $r\ge 2$.
Note that the claim applies to any graph $G_N$ and any fixed starting position $v$ of the invader.
If the graph $G_N$ is not a path graph then the inequality is strict from every starting node, implying that the graph is more robust with respect to replacer invasions than the natural baseline.
If the graph $G_N$ is a path graph then a replacer invader appearing at its endpoint has exactly the same fixation probability as the value given by the natural baseline.
However, for $N\ge 3$ there are other possible nodes where the invader could appear, and from there the fixation probability is lower.
Therefore, upon averaging over possible starting positions, we obtain a strict inequality for the path graph too, provided that $N\ge 3$. This is consistent with results plotted in~\cref{fig:all_vertices}, where the path graph $P_7$ maximizes the fixation probability among all graphs with 7 nodes, but is still below the baseline.

The idea behind the proof of the second claim relies on two ingredients.
For the first ingredient, one can show that when $G_N$ is the path graph and the initial replacer invader appears at its endpoint, then its fixation probability is exactly equal to the baseline expression $\frac{1-\frac{1}{r}}{1-\frac{1}{r^N}}$.
In fact, this claim is true not only for $r\ge 2$, but for any $r>0$. 
Second, we show that if $r\ge 2$ then the fixation probability from any initial node of any graph is at most as high as the fixation probability starting from an endpoint of a path graph.
To do this, in essence we run the process until the invaders first reach a moment when their subpopulation has at least two resident neighbors.
This puts them at a difficult spot, because if either of those two resident neighbors is selected for reproduction, one invader is replaced.
Thanks to the $r\ge 2$ assumption, this immediate vulnerability can be showed to outweigh any potential benefits that might come later in the process. See \SI{} for details.

\section{Conclusions}

The evolutionary dynamics in finite spatially structured populations is commonly modeled by variants of the Moran process.
In each step of the Moran process, one individual is reproducing, and the offspring replaces a neighbor.
In the standard Moran process, the reproducing individuals select the neighbor uniformly at random, so sometimes they effectively waste their reproductive turn by replacing an individual of the same type as the parent.
In this work, we consider the recently introduced phenotype called the replacer.
A replacer also replaces a random neighbor, but only from among those who have a different type than the reproducing parent (if  such neighbors are available).
Replacers are therefore more efficient competitors than the standard, oblivious individuals.
Past work studied the evolutionary dynamics between replacers and oblivious individuals~\cite{pecho2025selective}.
Here we study the evolutionary dynamics between replacers and replacers.
Namely, we analyze the fixation probability $\fps(G_N)$ of a single replacer with relative reproductive rate $r$, as it attempts to invade a homogeneous population of established replacers who are arranged at nodes of a graph $G_N$.

We find that established replacers are generally well protected against replacer invasion attempts,
and we quantify the strength of this effect for various spatial structures.
The effect is strongest for well-mixed populations, represented by complete graphs $K_N$, where we derive $\fps(K_N)=\frac{1}{\left(1+\frac1{r}\right)^{N-1}}$.
That is, the invader fixation probability is exponentially small for any fixed $r>0$.
Thus, large populations of replacers are extremely robust with respect to replacer invasions, even when the invader is advantageous ($r>1$).
An analogous result holds for other dense graphs such as island models with a constant number of islands.

At the other extreme, we have sparse graphs such as low-dimensional lattices.
For large cycle graph $C_N$, the fixation probability of an advantageous invader is decreased compared to the oblivious case, but remains constant for any $r>1$ rather than being exponentially small.
Thus, the effect of increased robustness persists, but it is substantially weaker.
For other sparse graphs, the magnitude of the effect depends not only on the degree, but also on the connectivity properties of the underlying graph. Many large 4-regular graphs are reasonably good expanders, and as such they are overwhelmingly likely to withstand invasions attempts of strength up to at least $r=1.9$, whereas for other graphs such as the Square grids $\Sq_N$ or the Bead graphs $B_N$ the same invasion attempts have a decent chance of being successful.

Finally, we present bounds for arbitrary graphs, showing that in the setting with replacers against replacers, the effect of increased robustness is in place for every single spatially structured population of any size, provided that the invader is either near-neutral ($r\approx 1$) or sufficiently strong ($r\ge 2$).

Being of the replacer phenotype generally helps both the invading mutant and the established residents.
However, our results show that the residents benefit more from this power.
This is particularly evident in the early stages of the process.
When only a single invader is present, it does not benefit at all from the extra power, since all its neighbors are anyway of the other type. In contrast, all those resident neighbors benefit greatly from the extra power, since the invader is only one of possibly many neighbors that they have.
Since the early stages of the invasion dynamics tend to be more important than the later stages, this asymmetry results in established populations being more resistant against replacer invasions, as compared with the baseline case of oblivious invaders and oblivious residents.

Our results go beyond this qualitative assessment and quantify the strength of the effect.
We show that high-degree graphs guarantee evolutionary stability against replacer invasions of any strength.
In contrast, for low-degree graphs specific features of the spatial structure become important as they determine whether the invasion attempt of a given strength $r>1$ is reasonably likely to succeed, or whether the established population is stable with respect to invasions of this strength.
A particularly important feature seems to be whether the graph is a good expander or not~\cite{alon1986eigenvalues}.

We also note that the effect of the spatial structure is major.
For neutral invaders ($r=1$) the fixation probability may be as low as $1/2^{N-1}$ (as is the case for the complete graph $K_N$) or as high as $1/(2N-2)$, as for the cycle graph $C_N$. Already for $N=20$ this is a difference of four orders of magnitude (as we have $\fpns(K_{20})<0.0002\%$, whereas $\fpns(C_{20})>2\%$).
This is in stark contrast to the standard case of oblivious invaders and residents, where the fixation probability of any neutral invader is equal to $1/N$ (so $5\%$ for $N=20$).
A similar contrast can be made for the case of advantageous mutants ($r>1$). Take $r=2$ and $N$ large.
In the standard case of oblivious invaders and residents, the Isothermal Theorem~\cite{lieberman2005evolutionary} implies that the fixation probability approaches $50\%$, for any regular graph $G_N$.
In contrast, in our case of the Moran process with replacers, for $r=2$ and $N=100$ we find
$\fps(C_{100})\approx 33\%$ for the cycle graph,
$\fps(\Sq_{100})\approx 3\%$ %
for the square grid, and
$\fps(K_{100})<10^{-10}$ for the complete graph.
By~\cref{thm:arbitrary_main}, all those values are less than the natural baseline of $50\%$, but the values span several orders of magnitude.

We conclude with stating two interesting future directions.
First, our result for arbitrary graphs shows that, compared to the natural baseline, the evolutionary stability of replacers is increased both when $r\approx 1$ and when $r\ge 2$. Is the same true for arbitrary $r$? Our numerical results for small $N$ and our simulation results for large $N$ suggest so.
However, proving the statement might not be easy, since a similarly sounding statement is actually not true when $r\ll 1$. Namely, there exists a (non-regular) graph $G_N$ and $r\ll 1$ such that $\fps(G_N)>\fp(G_N)$, see Remark~2 %
in \SI{}.
That is, when comparing the evolutionary dynamics of replacers to the evolutionary dynamics of oblivious individuals on the same spatial structure (rather than comparing it to the baseline given by any isothermal structure), on some graphs the evolutionary stability of replacers might actually slightly decrease when $r\ll 1$.
However, in that case both fixation probabilities are extremely small.

Second, fixation probability is just one representative quantity of the evolutionary dynamics.
Other relevant quantities include fixation and diversity times~\cite{broom2010evolutionary,tkadlec2019population,svoboda2023coexistence,bulmer1972multiple,yeaman2011establishment}, or the properties of the mutation-selection equilibrium~\cite{sharma2022suppressors,yagoobi2018mutation}.
In the original Moran process on arbitrary undirected graphs with $r>1$, those quantities are all relatively mild.
In particular, the fixation probability is always at least inversely proportional to the population size $N$,
the fixation time is at most quartic~\cite{brewster2025maintaining},
and there is a tradeoff between fixation probability and fixation time~\cite{tkadlec2021fast}.
In the Moran process with replacers, the fixation probability on some undirected spatial structures is exponentially small, so it is conceivable that on some structures the fixation time is exponentially long. Identifying such structures would be interesting, especially if they also lead to long coexistence and diversity times.

\section*{Code availability}
The data and code are available at \url{https://github.com/mpecho19/moran-process-replacers}.

\section*{Acknowledgments}
J.T. and M.P. were supported by Charles Univ.\ projects UNCE 24/SCI/008 and PRIMUS 24/SCI/012.
J.T. was supported by GA\v{C}R grant 25-17377S.
K.C. was partially supported by the Austrian Science Fund (FWF) 10.55776/COE12, and by the ERC CoG 863818 (ForM-SMArt) grant.

\bibliographystyle{naturemag}

\newpage
\setcounter{thm}{0}
\setcounter{section}{0}
\section*{Appendix}\label{sec:organization}

This is the Appendix for the paper \textit{\maintitle}.
It contains formal proofs of the theorems listed in the main text.

\section{Model}\label{sec:model}

Let $G=(V,E)$ be a finite, simple, connected, undirected graph with
$N=|V|$. Each vertex is occupied by one individual, which is either a
resident or a mutant. Residents have fitness $1$, while mutants have fitness $r>0$.
A state of the population is denoted by a set
$S\subseteq V$ of vertices occupied by mutants. Thus, $S=\varnothing$
corresponds to mutant extinction and $S=V$ to mutant fixation.

We consider discrete-time birth--death updating. In a state $S$, the total
population fitness is
\[
    F(S)=N-|S|+r|S|.
\]
An individual in $v$ is selected for reproduction with probability $\frac{r}{F(S)}$ if it is a mutant and with probability $\frac{1}{F(S)}$ if it is a resident.
The offspring inherits the type of its parent and replaces an individual at a neighbouring vertex.
We define the original process and a new process with neighborhood-aware individuals, called replacers.

\begin{defn}[Moran process, $\mor$]\label{def:moran}
In the standard Moran process with birth--death updating, denoted by $\mor$,
the reproducing vertex $v$ selects a vertex uniformly at random from its
neighbourhood
\[
    N(v)=\{u\in V:\{u,v\}\in E\}.
\]
The selected neighbour is then replaced by an offspring of $v$.
In particular, the update has no effect when the selected neighbour already has the same type
as $v$.
\end{defn}

\begin{defn}[Moran process with replacers, $\sen$]\label{def:replacers}
In the Moran process with replacers, denoted by $\sen$, the reproducing vertex
$v$ selects uniformly at random among its neighbours of the opposite type. For
a state $S$, define
\[
    N_{\mathrm{opp}}(v,S)
    =
    \begin{cases}
        N(v)\setminus S, & v\in S,\\
        N(v)\cap S,      & v\notin S.
    \end{cases}
\]
If $N_{\mathrm{opp}}(v,S)\neq\varnothing$, a vertex in this set is selected
uniformly and replaced by an offspring of $v$. If
$N_{\mathrm{opp}}(v,S)=\varnothing$, no replacement occurs and the state is
unchanged.
\end{defn}

Thus, unlike in $\mor$, a reproducing individual in $\sen$ is aware of the
types in its neighborhood and, whenever possible, replaces an individual of
the opposite type. This rule applies symmetrically to residents and mutants.

Because $G$ is connected, the only absorbing states of either process are
$S=\varnothing$ and $S=V$. For $\mathcal P\in\{\mor,\sen\}$, define the so-called \textit{fixation probability}
\[
    \rho_r^{\mathcal P}(G;S)
    =
    \Pr^{\mathcal P}\!\left(
        \text{the process reaches }V
        \,\middle|\,
        S_0=S
    \right).
\]
For a state with a single mutant in a vertex $v$ we write 
\[\rho_r^{\mathcal P}(G;v) = \rho_r^{\mathcal P}(G;\{v\}). \]
Unless stated otherwise, the process starts with a single mutant placed
uniformly at random. We therefore write
\[
    \fp(G)
    =
    \frac{1}{N}\sum_{v\in V}
    \rho^{\mor}(G;\{v\}),
    \qquad
    \fps(G)
    =
    \frac{1}{N}\sum_{v\in V}
    \rho^{\sen}(G;\{v\}).
\]

\paragraph{Well-mixed baseline.}
To compare the ability of residents to withstand the invasion of mutants in the replacers regime, we compare the fixation probability of mutants with the fixation probability in a well-mixed population (denoted $K_N$) under
the standard Moran process~\cite{moran1958random}.
The fixation probability from a single mutant is
\begin{equation}\label{eq:moran-baseline}
    \fp(K_N)
    =
    \begin{cases}
        \displaystyle
        \frac{1-r^{-1}}{1-r^{-N}},
        & r\neq 1,\\[10pt]
        \displaystyle
        \frac{1}{N},
        & r=1.
    \end{cases}
\end{equation}

\begin{defn}[Backward bias]\label{def:gamma}
Let $p_k$ be the probability transition probability from state $k$ to state $k+1$ and 
$q_k$ the probability transition probability from state $k$ to state $k-1$. We define
\[
\gamma_k = \frac{q_k}{p_k}.
\]
\end{defn}

The following lemma is standard, see e.g.~\cite[Eqn 6.12]{nowak2007evolutionary}.
\begin{lemma}[One-dimensional random walk] %
\label{lem:one-dimensional} 
Consider a one-dimensional random walk with states $\{0,1,\dots,N\}$ and absorbing boundaries at $0$ and $N$. 
Given $\gamma_1,\dots,\gamma_{N-1}$ as in Definition~\ref{def:gamma}, the probability of absorption at $N$ 
starting from state $i$ is
\[
\rho = \frac{1+\sum_{j=1}^{i-1}\prod_{k=1}^{j}\gamma_k}{1 + \sum_{j=1}^{N-1}\prod_{k=1}^{j}\gamma_k}.
\]
\end{lemma}

In this paper we mainly use \cref{lem:one-dimensional} when working with symmetric graphs with states representing different number of mutants and the transition probabilities $p_k$ and $q_k$ representing the probability of increasing and decreasing the number of mutants. We also use the following notation.

\begin{defn}
    We define $\gamma(S)$ as the ratio of the probability of losing a mutant to the probability of gaining a mutant in a state with set of mutants $S$. 
\end{defn}

\begin{defn}
    Let $G$ be a graph and $S$ is a set of mutants in a given state, then we call an individual $u$ active, if there exist an individual $v$ of the opposite type, such that there is an edge between $u$ and $v$ in $G$. We define $\AM(S)$ as the number of active mutants and $\AR(S)$ as the number of active residents.
\end{defn}

\begin{rem}
    When an active individual is reproducing in $\sen$, it replaces some individual of the opposite type with probability $1$, therefore given a state with set of mutants $S$ we can compute

$$\gamma(S) =\frac{\frac{\AR(S)}{|S|\cdot r + (N-|S|)}}{ \frac{r\cdot\AM(S)}{|S|\cdot r + (N-|S|)}}= \frac{1}{r}\cdot \frac{\AR(S)}{\AM(S)}.$$

\end{rem} 
\section{Fixation probability in a well-mixed population}

The first question about the process is the fixation probability of invading mutants of any fitness in a well-mixed population.
We show the fixation probability for all starting mutant sizes and all $r>0$ and then use it to show the fixation probability of a single mutant. Then we show that if the minimal degree is high, then the fixation probability is exponentially small in the size of the degree.

\begin{lemma}\label{lem:ss_exp_small}
    Let $r>0$. For a complete graph $K_N$, the fixation probability of initial mutants occupying set $S$ with size $|S|=i$ satisfies
    \[
        \fps(K_N;S) = \frac{1 + \sum_{j=1}^{i-1}\binom{N-1}{j}\cdot\frac{1}{r^j}}{(1+ \frac1r)^{N-1}}= \mathbb{P}(X\leq i-1),
    \]
    where $X\sim \operatorname{Binomial}(N-1, \frac{1}{r+1})$.
\end{lemma}

\begin{proof}
    On a complete graph, the probability of gaining a mutant is the probability of choosing one of the $k$ mutants for reproduction and then choosing a resident.
    The probability of choosing a mutant for reproduction is $\frac{r\cdot k}{T_k}$ and as mutants have the sensing ability, they will choose a resident with probability $1$. 
    Hence, $p_k = r\cdot \frac{k}{T_k}\cdot 1$. 
    
    Similarly, the probability of losing a mutant is the probability of choosing one of the $N-k$ residents for reproduction and then choosing one of the $k$ mutants. The probability of choosing a resident for reproduction is $\frac{N-k}{T_k}$ and as residents have the sensing ability, they will choose a mutant with probability $1$. Hence, $q_k = \frac{N-k}{T_k}\cdot 1$.
    
    Then $\gamma_k = \frac{q_k}{p_k} = \frac{N-k}{k\cdot r}$.
    We use \Cref{lem:one-dimensional} to compute the fixation probability. We have that 
    \[
        \prod_{k=1}^{j} \gamma_k = \frac{(N-1)\cdot (N-2) \cdots (N-j)}{j!\cdot r^j}= \frac{(N-1)!}{j!(N-j-1)!}\cdot\frac{1}{r^j} =  \binom{N-1}{j}\cdot\frac1{r^j}. 
    \] 

    Therefore, the fixation probability for $i$ starting mutants in set $S$ is equal to
    
    \[\fps(K_N;S) = \frac{1 +\sum_{j=1}^{i-1}\binom{N-1}{j}\cdot\frac{1}{r^j}}{1 + \sum_{j=1}^{N-1}\binom{N-1}{j}\cdot\frac{1}{r^j}} = \frac{1 + \sum_{j=1}^{i-1}\binom{N-1}{j}\cdot\frac{1}{r^j}}{(1+ \frac1r)^{N-1}}, \]
    where we used the binomial expansion. 

    Next, we have that 
    \[
    \fps(K_N;S) = \frac{\sum_{j=0}^{i-1} \binom{N-1}{j}r^{-j}}{(1+\frac1r)^{N-1}}= \frac{\sum_{j=0}^{i-1} \binom{N-1}{j}\frac{1}{(1+r)^j}\frac{r^{-j}}{(1+r)^{-j}}}{(\frac{r+1}r)^{N-1}}=\frac{\sum_{j=0}^{i-1} \binom{N-1}{j}\frac{1}{(1+r)^j}(1-\frac{1}{1+r})^{-j}}{(1-\frac{1}{1+r})^{-(N-1)}}\]
    \[=\sum_{j=0}^{i-1}\binom{N-1}{j}\left(\frac{1}{1+r}\right)^j \left(1-\frac{1}{1+r}\right)^{N-1-j}= \mathbb{P}(X\leq i-1),
    \]
    where $X \sim \operatorname{Binomial}(N-1, \frac{1}{r+1})$.
\end{proof}

\begin{corollary}
    Let $r>0$ and $\varepsilon >0$. For a complete graph $K_N$, the fixation probability of initial mutants occupying set $S$ with size $|S|=i < N(\frac{1}{1+r} - \varepsilon)$ satisfies

    \[
    \fps(K_N;S) \leq \exp(-2(N-1)\varepsilon^2).
    \]
\end{corollary}
\begin{proof}
From \cref{lem:ss_exp_small} we have that 
\[\fps(K_N) = \mathbb{P}(X\leq i-1),
\]
where $X\sim \text{Binomial}(N-1, \frac{1}{r+1})$. We can write $X = \sum_{k=1}^{N-1}Y_k$, where $Y_k$ are random variables following Bernoulli distribution, $\mathbb{P}(Y_k =1) = \frac{1}{r+1}$ and $\mathbb{P}(Y_k= 0) = 1-\frac{1}{1+r}$, with mean $\frac{1}{r+1}$. Then we have that 

\[
 \mathbb{P}(X\leq i-1) \leq \mathbb{P}\left(X\leq (N-1)\cdot \left(\frac{1}{1+r} - \varepsilon\right)\right) = \mathbb{P}\left(\frac{1}{1+r} -\frac{\sum_{k=1}^{N-1} Y_k}{N-1} \geq \varepsilon\right) = \mathbb{P}\left(\frac{\sum_{k=1}^{N-1}-Y_k}{N-1} - (-\frac{1}{1+r}) \geq \varepsilon\right).
\]

We have that $-1 \leq -Y_k \leq 0$ and the mean of $-Y_k$ is equal to $-\frac{1}{1+r}$, therefore we can use Hoeffding's inequality \cite[Theorem 2]{Hoeffding1963}
\[
\fps(K_N) = \mathbb{P}(X\leq i-1) \leq \mathbb{P}\left(\frac{\sum_{k=1}^{N-1}-Y_k}{N-1} - \left(-\frac{1}{1+r}\right) \geq \varepsilon\right) \leq \exp\left(-\frac{2(N-1)^2\varepsilon^2}{(N-1)\cdot 1^2}\right) = \exp\left(-{2(N-1)\varepsilon^2}\right).
\]
\end{proof}

By plugging in $i=1$ we obtain the claim stated in the main text.
\begin{corollary}\label{cor:one_mut_well_mixed}
    For $K_N$ and $r > 0$, we have
    \[\fps(K_N)= \frac{1}{(1+ \frac1r)^{N-1}}.\]
\end{corollary}

\begin{lemma}\label{lem:min_degree}
    Let $G_N$ be a graph with $N$ vertices and minimum degree $D$. Then the fixation probability satisfies
    \[\fps(G_N) \leq \frac{1}{(1+\frac1r)^D}.\]
\end{lemma}
\begin{myproof}
    Let $M$ be a set of mutants of size $k$. Then we have that $\AM(M) \leq k$. Next, we look at a single mutant. The mutant has at least $D$ neighbors and at most $k-1$ of them are mutants. Therefore, we have that $\AR(M) \geq D-(k-1)$. Hence we can bound the backward bias 

    $$\gamma(M) = \frac1r \cdot \frac{\AR(M)}{\AM(M)} \geq \frac{D-k+1}{r\cdot k}.$$

    We use \cref{lem:one-dimensional} to bound the fixation probability 

    $$\fps(G_N) \leq \frac{1}{ 1+\sum_{j=1}^{D}\prod_{k=1}^{j} \frac{D-k+1}{k}\cdot \frac{1}{r}}=\frac{1}{1+ \sum_{j=1}^{D}\frac{D!}{j!\cdot (D-j)!}\cdot \frac{1}{r^j}}=\frac{1}{\sum_{j=0}^D\binom{D}{j}\frac{1}{r^j}}=\frac{1}{(1+\frac{1}{r})^D}.$$
\end{myproof}

Combining \cref{cor:one_mut_well_mixed} and \cref{lem:min_degree} we get the following theorem.

\begin{thm}\label{thm:1}
\leavevmode
\begin{enumerate}
    \item (Complete graph) Let $K_N$ be a complete graph and $r>0$. Then
\[
\fps(K_N) = \frac{1}{\left(1+\frac1r\right)^{N-1}}.
\]
\item (High-degree graphs) Let $G_N$ be any graph with minimum degree $D$. Then
\[
\fps(G_N) \leq \frac{1}{(1+\frac1r)^D}.
\]
\end{enumerate}
\end{thm} 
\section{Fixation probability in a cycle}

In this section, we derive explicit formulas in a simple one-dimensional structure represented by a cycle graph $C_N$.
We derive the fixation probability for $i$ initial mutants forming a contiguous block and for all $r>0$.
Then, we use it to show the fixation probability of a single mutant. 
Finally, we show the asymptotic behavior of the formula for large $N$. 

\begin{lemma}\label{lem:cycle-ss}
    Let $r>0$ and $N \ge 3$. For a cycle $C_N$, the fixation probability of initial mutants forming a block $S$ of size $i <N$ satisfies
    \[ \fps (C_N; S) =\begin{cases}
     \frac{2i -1}{2N- 2} \text{\ \ if $r = 1$,}\\
     \\
     \frac{1 + \frac{1}{r} - \frac{2}{r^i}}{1+\frac1r -\frac{1}{r^{N-1}} - \frac{1}{r^{N}}} \text{\ \ if $r\not = 1$.}
 \end{cases}\]
\end{lemma}

\begin{proof}
    Starting from a contiguous block of mutants, throughout the process, mutants occupy a contiguous block of vertices.
    Let $k$ be the number of mutants in the block.
    Note that $p_1 = \frac{r}{T_1}$.
    For $k \geq 2$, we have $p_k = \frac{2r}{T_k}$, as there is one active mutant at each end of the block.
    Similarly, $q_{N-1} = \frac{1}{T_{N-1}}$.
    For $k\leq N-2$, the residents will form a block, and there are two active residents; hence $p_{i}^- = \frac{2}{T_i}$.
    Therefore $\gamma_1 = \frac{2}{r}$, $\gamma_{N-1}=\frac{1}{2r}$ and $\gamma_k = \frac{1}{r}$ for $2 \leq k \leq N-2$. Thus, by~\Cref{lem:one-dimensional} we get

\begin{align*}
    \fps(C_N;S) &= \frac{1+\sum_{j=1}^{i-1}\prod_{k=1}^{j}\gamma_k}{1+ \sum_{j=1}^{N-1}\prod_{k=1}^j\gamma_k }\\
    &=\frac{2(1+ \sum_{j=1}^{i-1}\frac{1}{r^j}) - 1 }{2(1 + \sum_{j=1}^{N-1}\frac{1}{r^j})-1 - \frac{1}{r^{N-1}}}. 
\end{align*}
We consider two cases. First, $r=1$. Then
\begin{align*}
     \fpns(C_N; S) &=\frac{2(1+ \sum_{j=1}^{i-1}\frac{1}{r^j}) - 1 }{2(1 + \sum_{j=1}^{N-1}\frac{1}{r^j})-1 - \frac{1}{r^{N-1}}} \\
     &=\frac{2i-1}{2N-2}.
\end{align*}
If $r\not= 1$, then 
\begin{align*}
\fps(C_N; S) &= \frac{2(1-\frac{1}{r^i})(1-\frac{1}{r})^{-1} - 1  }{2(1-\frac1{r^{N}})(1-\frac1r)^{-1} -1 -\frac{1}{r^{N-1}}} \\
    &=  \frac{1 + \frac{1}{r} - \frac{2}{r^i}}{1+\frac1r -\frac{1}{r^{N-1}} - \frac{1}{r^{N}}}.
\end{align*}
\end{proof}

\begin{corollary}
    Let $r>0$, then for a single initial mutant we have 
    $$ \fps (C_N) =\begin{cases}
     \frac{1}{2N- 2} \text{\ \ if $r = 1$,}\\
     \\
     \frac{1 - \frac{1}{r}}{1+\frac1r -\frac{1}{r^{N-1}} - \frac{1}{r^{N}}} \text{\ \ if $r\not = 1$.}
     \end{cases}$$
\end{corollary}
We also derive asymptotics for the fixation probability of a single initial mutant on a cycle. 
\begin{corollary}\label{cor:cor-cycle-ss}
    Let $C_N$ be the cycle of size $N$, where $N$ is large. Then
\begin{align*}
\fpns(C_N) &= \frac1{2N-2}
\text{\ \  if $r=1$, }\\
\fps(C_N) &\approx \frac{1-\frac1r}{1+\frac1r}
\text{\ \  if $r>1$, }\\
\fps(C_N) &\sim r^N  
\text{\ \  if $r <1$}.
\end{align*}
\end{corollary}

\begin{proof}
    We have already proven the result for $r=1$ in \Cref{lem:cycle-ss}. Let $r\not=1$. Then according to \Cref{lem:cycle-ss} the fixation probability satisfies
    \[
     \fps(C_N) = \frac{1 - \frac{1}{r}}{1+\frac1r -\frac{1}{r^{N-1}} - \frac{1}{r^{N}}}.
    \]
    Let $r>1$, then 
    \[
    \fps(C_N) = \frac{1 - \frac{1}{r}}{1+\frac1r -\frac{1}{r^{N-1}} - \frac{1}{r^{N}}}\approx\frac{1-\frac1r}{1+\frac1r}.    \]
    Let $r<1$, then 
    \[
    \fps(C_N) = \frac{1 - \frac{1}{r}}{1+\frac1r -\frac{1}{r^{N-1}} - \frac{1}{r^{N}}} = r^N\cdot \frac{1-\frac1r}{r^N+r^{N-1} - r - 1}\approx r^N\cdot\frac{1-\frac1r}{-r-1}\sim r^N.
    \]
    
\end{proof} 
\section{Bounded-degree graphs}

In this section, we examine graphs with respect to their degree.
We present three types of results.
First, we show two general results: if the fitness is higher than the maximum degree, then the fixation probability is at least constant.
Second, we design a graph which we call an $\eps$-bead graph, where the fixation probability is constant for every $r>1$.
Finally, we demonstrate that a small maximal degree does not necessarily guarantee a large fixation probability:
we prove that the fixation probability on $4$-regular expanders is exponentially small for $1 \le r \le 1.9$.

\begin{lemma}
   Let $G_N$ be a graph with $N$ vertices and maximum degree $\Delta$.
   Suppose that $r\geq \Delta(1+\varepsilon)$ for some $\varepsilon >0$. Then the fixation probability satisfies 

    $$
    \fps(G_N) > \frac{\varepsilon}{1+\varepsilon}.
    $$

\end{lemma}
\begin{myproof}
    Let $M$ be a set of mutants of size $k$.
    Every active mutant has at most $\Delta$ resident neighbors.
    Therefore, in every step, we have
    \[
        \AR(M) \leq \Delta \cdot \AM(M).
    \]
    
    \[
        \gamma(M) =   \frac1r\cdot \frac{\AR(M)}{\AM(M)} \leq \frac{\Delta}{r} \leq \frac{1}{1+\varepsilon}.
    \]

    As the backward bias is at most $\frac{1}{1+\varepsilon}$ in every configuration, we get by \cref{lem:one-dimensional}.
    for $r\geq \Delta \cdot (1+\varepsilon)$, the fixation probability satisfies 
    \[
    \fps(G_N) \geq \frac{1-\frac{1}{1+\varepsilon}}{1-\frac{1}{(1+\varepsilon)^{N}}} > \frac{\varepsilon}{1+\varepsilon}.
    \]
\end{myproof}

Next, we show that there exist large 4-regular graphs, namely the bead graphs, where the fixation probability of only mildly advantageous replacer invaders is already substantial.

\begin{lemma}\label{lem:bead-constant-fix}
    There exists a constant $c>0$ such that for every $r\ge 1.1$ and every 4-regular graph $B_N$ we have
    $$\fps(B_N)>c.$$

\end{lemma}
\begin{proof}
    Fix $r\ge 1.1$. We want to bound the mutant fixation probability from below by a constant, independent of $N$. The argument consists of two parts.
    In the first part, we will argue that with at least a constant probability $c_0$ the mutant takes over the bead where it initially appeared.
    In the second part, we will show that from that point on, the mutants have at least $c_1$-times higher chance to take over a new bead than to lose a bead, where $c_1>1$ is a constant strictly larger than 1.
    By \cref{lem:one-dimensional} we can then conclude that $\fps(B_N)\ge c_0\cdot (1-1/c_1)$, which is the required strictly positive constant. It remains to prove the two parts.

    Recall that a vertex in a current state of the population is \textit{active}, if it has a neighbor of the other type. Similarly, we say that a step of the process is \textit{active} if the set $S$ of vertices occupied by mutants changes in that step.

    For the first part, fix a sequence $(a_1,a_2,a_3,a_4)$ of four active steps that leads to the mutants taking over their initial bead, one vertex at a time. We argue that, ignoring the non-active steps, this sequence happens with constant probability. At any point during this sequence, at most 7 vertices are active (the five vertices in the bead, plus the two adjacent vertices) and at most 5 of those 7 vertices are mutants, so with probability at least $\frac{r}{5r+2}\ge \frac17$ the correct mutant vertex gets selected for reproduction. Since the mutant has only 4 neighbors, with probability at least $\frac14$ it replaces the correct neighbor, as prescribed by the sequence. Thus, with probability at least $\left(\frac17\cdot\frac14\right)^4=c_0$ the first four active steps precisely match the prescribed sequence, and the mutants win the initial bead.

    For the second part, we track the interface between the mutants and residents. We distinguish the \textit{clustered} states where each bead consists of individuals of only one type, and the \textit{mixed} states where one bead contains individuals of both types. The clustered state can be labeled by the number $k$ of (consecutive) beads that are fully occupied by mutants. When currently at state $k$, the next active step takes us either to state $k^+$, where mutants invade the next bead, or to state $k^-$, where residents invade the mutant bead, see~\cref{fig:beads}.
    The corresponding transition probabilities are $p^+=\frac{r}{r+1}> 52\%$ and $p^-=\frac{1}{r+1}< 48\%$, where we have used $r\ge 1.1$.

\begin{figure}
    \centering
    \includegraphics[width=0.8\linewidth]{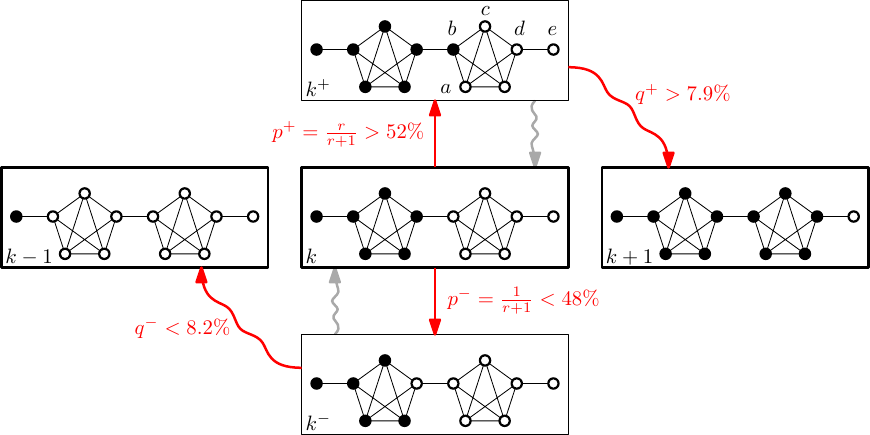}
    \caption{Evolutionary dynamics on the bead graphs.} Two adjacent beads are shown. States are labeled 
    \label{fig:beads}
\end{figure}
    
    To treat the mixed states, we consider a modified process $\sen'$ that is strictly harder for the mutants than $\sen$, and we argue that even in this modified process the chance of winning a bead is at least $c_1$-times larger than the chance of losing a bead. Namely, we modify the process $\sen$ in any mixed state as follows:
    \begin{enumerate}
        \item Whenever a mutant in a mixed bead attempts to place an offspring into an adjacent resident bead, that step is ignored. 
        \item Whenever a resident in a mixed bead places an offspring onto the vertex that is connected to the adjacent mutant bead, the current bead becomes all residents.
    \end{enumerate}
    For example, consider the state $k^+$, with vertices labeled $a$, $b$, $c$, $d$, $e$ as shown. If the next two active steps are $b\to c$ and $c\to d$, then in $\sen'$ the vertices occupied by mutants are $\{b,c,d\}$, mimicking the process $\sen$.
    If the next active step were $d\to e$, then in $\sen'$ this step is ignored.
    If the next active step were $a\to b$, then in $\sen'$ the whole bead becomes all residents and we return to state $k$.

    It is clear that with probability 1, the process $\sen'$ resolves in finite time to either the state where the whole mixed bead is mutants, that is, state $k+1$, or back to the state $k$. The underlying Markov chain has $2^5=32$ states, one for each possible configuration of mutants and resident on the bead. Solving the absorption probabilities on the Markov chain explicitly by inverting the corresponding $32\times 32$ matrix, we find that even with $r=1.1$ and even in $\sen'$, the probability of advancing to state $k+1$ (as opposed to reverting to state $k$) is at least $q^+>7.9\%$. See the linked code for details.

    Similarly, we find that once in state $k^-$, the probability of falling to state $k-1$ as opposed to reverting back to state $k$ (in $\sen'$, with $r=1.1$) is at most $q^-<8.2\%$.
    Thus, starting from state $k$, the probability $\gamma$ of visiting state $k+1$ before visiting state $k-1$ satisfies
    \[\gamma= \frac{p^+\cdot q^+}{p^-\cdot q^-}\ge\frac{52\cdot 7.9}{48\cdot 8.2}>1.04\]
    and we can set $c_2=1.04$.
\end{proof}

Finally, we define a vertex expander and show that for all $D \geq 4$ and $r\leq D-2.1$ there exists a $D$-regular graph for which the fixation probability of mutants is exponentially small. 

\begin{defn}
    A graph $G$ is a $(K,A)$ vertex expander if for all sets $S$ of at most $K$ vertices, the neighborhood $N(S)$ is of size at least $A\cdot |S|$.
\end{defn}

\begin{lemma}[\cite{expanders}]\label{thm:regular_expanders} %
    For all constants $D \geq 3$, there is a constant $\alpha > 0$ such that for all $N$, a random $D$-regular undirected graph on $N$ vertices is an $(\alpha N, D-1.01)$ vertex expander with probability at least $1/2$.
\end{lemma}

\begin{lemma}\label{lem:expander_exp_small}
    Let $D \geq 4$ be a positive integer and let $r \le D-2.1$.
    Then there exists $N_0$, such that for all $N \ge N_0$ there exists a $D$-regular graph $G_N$ with $N$ vertices such that 
    \[
        \fps(G_N)\le 2^{-\Omega(N)}
    \]
\end{lemma}
\begin{myproof}
    Let $\alpha >0$ be the constant from \Cref{thm:regular_expanders} and let $G_N$ be an $(\alpha N, D-1.01)$ expander.
    Let $M$ be the set of mutant vertices, then $\AM(M) \leq |M|$ and $\AR(M) \geq |N(M)| - |M|$.
    
    Suppose that  $|M| < \alpha N$. 
    Then for every such set of vertices $M$ we have $|N(M)| \geq (D-1.01)\cdot|M| $.
    Hence $\AR(M) \geq |N(M)| -|M| \geq (D-1.01)\cdot|M| - |M| \geq (D - 2.01)\cdot \AM(M) $.
    
    Therefore, in every state where the number of mutants is less than $\alpha N$, we have $\gamma(S) = \frac1r\cdot\frac{\AR(M)}{\AM(M)} \geq \frac1r\cdot (D-2.01)$.
    Notice that for $r \leq D -2.1$ we have 
    \[
    \gamma(S) \geq \frac1r\cdot ({D-2.01}) \geq \frac{D-2.01}{D-2.1},
    \]
    so there exists a constant $f >  1$ such that $\gamma(S) \geq f$.
    This means, until there are at most $\alpha N$ mutants, there is always bias against them.
    We can bound the fixation probability by:
    \[
        \fps(G_N) \leq \frac{1}{1+f + \cdots + f^{N\alpha-1}} = \frac{1 - f}{1-f^{N\alpha}}\,,
    \]
    which means $\fps(G_N)\le 2^{-\Omega(N)}$.
\end{myproof}

Combining \cref{cor:cor-cycle-ss}, \cref{lem:bead-constant-fix} and \cref{lem:expander_exp_small} we get the following theorem.

\begin{thm}\label{thm:cn}
\leavevmode
\begin{enumerate}
    \item (Cycle graph) Let $C_N$ be a cycle graph. Then    %
    \begin{enumerate}
        \item %
         If $r=1$ then $\fps(C_N)=\frac1{2N-2}$.
        \item %
         If $r>1$ and $N$ is large then
        $\fps(C_N) \approx\frac{1-\frac1r}{1+\frac1r}$.
        \item %
         If $r<1$ and $N$ is large then $\fps(C_N)\sim r^N$.
    \end{enumerate}
    
    \item (4-regular graphs)
    \begin{enumerate}
        \item For every $r\ge 1.1$ there exists a positive constant $c$ and large 4-regular graphs $B_N$ such that $\fps(B_N)>c$.
        \item For every $r\le 1.9$ there exists a positive constant $c>1$ and large 4-regular graphs $E_N$ such that $\fps(E_N)<1/c^N$.
    \end{enumerate}
    
    \end{enumerate}
\end{thm} 
\section{Comparison to baseline}
In this section, we compare the fixation probability of replacers with the well-mixed baseline in \Cref{eq:moran-baseline}. %

\begin{lemma}\label{lem:comparison_r_1}
 Let $G_N$ be any graph with $N\geq 3$ nodes. Then $\fpns(G_N)<\frac{1}{N}=\fpn(G_N)$.
\end{lemma}
\begin{myproof}
Let us consider the process in the smart scenario starting from a mutant vertex $v$.
Let us consider a second process in which all the nodes have a different color and all follow the smart strategy. Then it is clear that starting from $v$ in the first case has a fixation probability not greater than in the second process. If we fix a subset of vertices that are mutants in some given state, the probability of gaining a new mutant is the same in both cases. However, losing a mutant is always greater than or equally probable in the first case than in the second case (because a resident may replace a neighboring resident of a different color). Hence, the first process has to have a fixation probability less than or equal to the second one. 

This holds for every starting position $v$, therefore it has to hold also when we average over all starting positions. But because the second process is symmetric, the fixation probability starting from a uniformly random vertex $v$ is $\frac{1}{N}$, which is the same as $\fpn(G_N).$ And hence, $\fpns(G_N)\leq\frac{1}{N}=\fpn(G_N)$. Moreover, for $N\geq3$ the inequality is actually strict because in at least one of the cases we average over, the starting state contains an active resident neighboring with a resident of a different color.
\end{myproof}

\begin{corollary}\label{cor:fixation_near_1}
    Let $G_N$ be any graph with $N \geq 3$ nodes. Then there exist $\varepsilon > 0$ such that for all $r\in (1-\varepsilon, 1+\varepsilon)$ we have $\fps(G_N) < \frac{1}{N}$.
\end{corollary}

\begin{lemma}\label{lem:path_comparison}
    Let $G_N$ be any graph with $N\ge 2$ nodes and $v$ a node of $G_N$. Let $r \geq 2$. Then
    \[
        \fps(G_N, v) \le \frac{1-\frac{1}{r}}{1-\frac{1}{r^N}}\,.
    \]
    Moreover, the equality occurs if and only if $G_N$ is the path graph and $v$ is one of its endpoints.
\end{lemma}

\begin{myproof}
    Suppose that the vertex $v$ is of degree $d\geq 2$. Then the probability of dying in the first step is $\frac{d}{d+r} \geq \frac{2}{2+r}$. Note that for $r\geq2$ is $\frac{2}{2+r} \geq \frac{1}{r}$. Therefore $\fps(G_n;v) \leq 1 - \frac{2}{2+r} \leq 1-\frac{1}{r}$

    Suppose that $G_N$ is a graph that is not a path and contains vertices of degree $1$. Then there exist a vertex of degree at least $3$. Suppose that $v$ is a vertex of degree $1$. Then $v$ is at the end of some sub-path and we may follow this sub-path until we reach a vertex of degree at least $3$. Let us call this vertex $u$ and the sub-path between them as $P_k$. Let $A$ be the set of neighbors of $u$ that are in $G_n$ but not in $P_k$. Then the probability that mutants will eventually get a vertex in $A$ is equal to 
    \[\frac{1}{1+\frac{1}{r} + \cdots + \frac{1}{r^{k-1}} + (d-1)\cdot \frac{1}{r^k}}.\]
    Note that $(d-1) \geq 2$ and for $r\geq2$ we have $\frac1{r^k}\geq\sum_{i=k+1}^\infty \frac1{r^{i}}$, hence
    \[\frac{1}{1+\frac{1}{r} + \cdots + \frac{1}{r^{k-1}} + (d-1)\cdot \frac{1}{r^k}} \leq \frac{1}{\sum_{i=0}^\infty \frac1{r^{i}}}= 1-\frac{1}{r}.\]
    
    Therefore for any starting vertex $v$, except for the end of a path, the fixation probability is bounded $\fps(G_n;v) \leq 1- \frac{1}{r}$ and we have \[\fps(G_n) \leq 1 - \frac{1}{r} < \frac{1-\frac1r}{1-\frac1{r^N}}.\]

    Finally, consider that $G_N$ is a path on $N$ vertices and $v$ is the end of the path, then throughout the process there will be always one active mutant and one active until the process ends. Therefore the backward bias will be $\gamma_k = 1/r$ for any $1\leq k \leq N-1$. By \cref{lem:one-dimensional} the fixation probability is equal to 
    \[\fps(G_N; v) = \frac{1-\frac1r}{1-\frac1{r^N}}.
    \]
\end{myproof}

Combing \cref{lem:comparison_r_1} and \cref{lem:path_comparison} we get the following theorem. 

\begin{thm}\label{thm:arbitrary}
\leavevmode
\begin{enumerate}
    \item ($r= 1$)
    Let $G_N$ be any graph with $N\geq 3$ nodes. Then $\fpns(G_N)<\frac{1}{N}=\fpn(G_N)$.
    \item ($r\ge 2$) Let $G_N$ be any graph with $N\ge 2$ nodes and $v$ a node of $G_N$. Let $r \geq 2$. Then
    \[
        \fps(G_N, v) \le \frac{1-\frac{1}{r}}{1-\frac{1}{r^N}}\,.
    \]
    Moreover, the equality occurs if and only if $G_N$ is the path graph and $v$ is one of its endpoints.
    \end{enumerate}
\end{thm}

\begin{rem}
It is not true that the fixation probability in $\mor$ is always better than the fixation probability in $\sen$. First, we note that there exist graphs that are better in $\sen$ than in $\mor$ in the setting $r<1$ (for example star on 4 vertices). Also, it is not true that for $r>1$ the $\sen$ is always worse from a particular starting position. As a counterexample might serve a path on 3 vertices and the middle vertex as the starting position. We then get this random walk:

\begin{figure}[!hbt]
\centering
\includegraphics[scale=0.8]{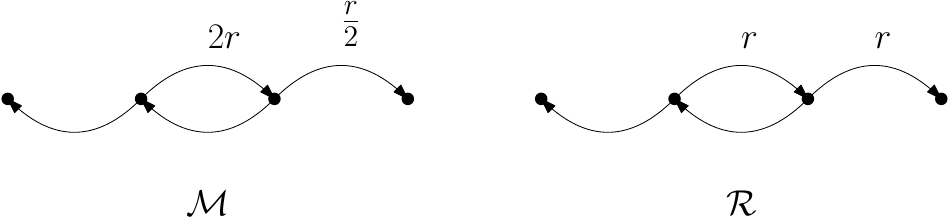}
\end{figure}

The first step is independent of the process. After the first step, mutants either die out or occupy the middle and one end of the path. This state is more beneficial for SS mutants. We can compute the extinction probability starting from this state to get $\frac{1}{1+\frac{r}{2}+r^2}$ for $\mor$ and $\frac{1}{1+r+r^2}$ for $\sen$. Hence, the fixation probability is greater for $\sen$ starting in the middle vertex for any $r>0.$
\end{rem}

\section{Graphs on $7$ vertices}

This section contains extra figure containing numerically computed fixation probability on all connected graphs on $7$ vertices for $r \in \{0.9, 2.0\}$.

\begin{figure}[h!]
    \centering
    \includegraphics[width=0.49\linewidth]{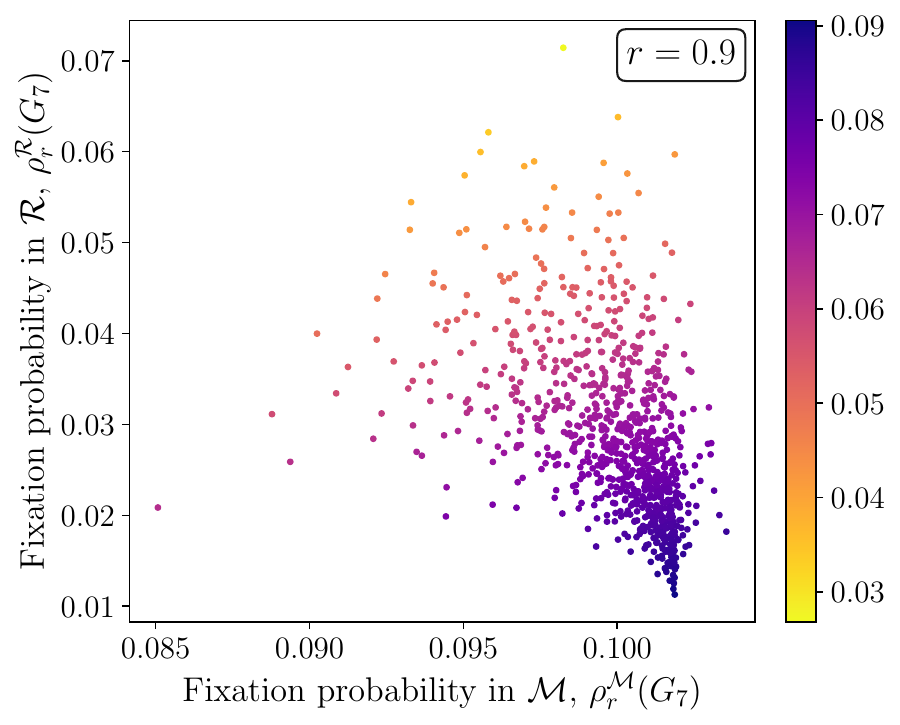}
    \includegraphics[width=0.49\linewidth]{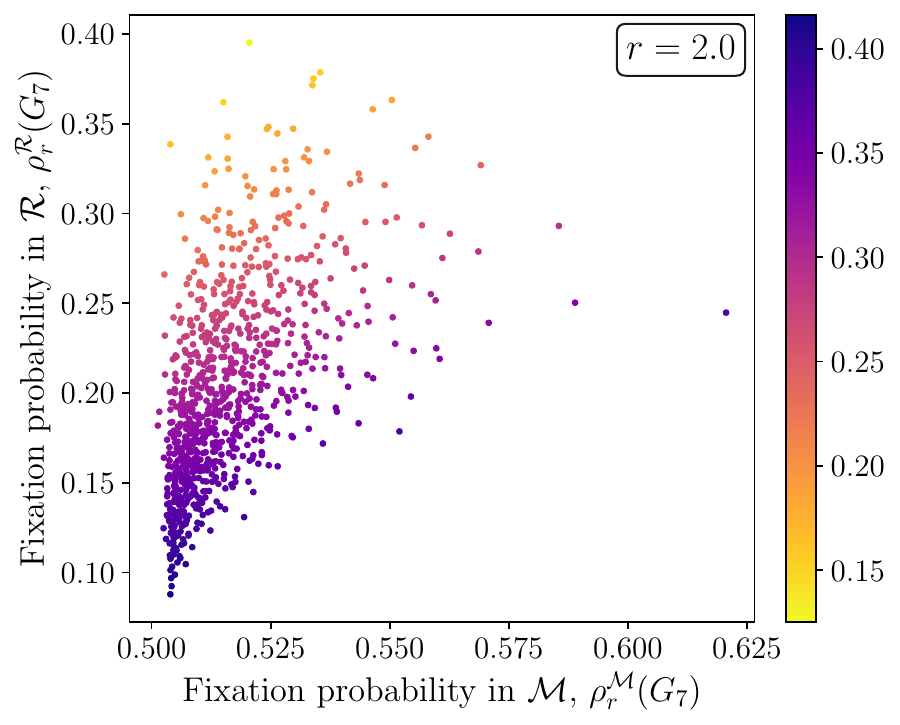}
    \label{fig:placeholder}
\end{figure}

\end{document}